\documentclass[letterpaper, 10 pt, conference]{ieeeconf}  

\IEEEoverridecommandlockouts                              

\usepackage[utf8]{inputenc}
\usepackage{include_packages} 
\usepackage[english]{babel}

\title{Carrots or Sticks? \\ The Effectiveness of Subsidies and Tolls in Congestion Games}
\author{\authorblockN{Bryce L. Ferguson}
\and
\authorblockN{Philip N. Brown}
\and
\authorblockN{Jason R. Marden}
\thanks{This research was supported by ONR grant \#N00014-17-1-2060 and NSF grant \#ECCS-1638214}
\thanks{B. L. Ferguson (corresponding author) and J. R. Marden are with the Department of Electrical and Computer Engineering, University of California, Santa Barbara, CA, {\texttt{\{blferguson,jrmarden\}@ece.ucsb.edu}}.}
\thanks{P. N. Brown is with the Department of Computer Science, University of Colorado at Colorado Springs, {\texttt{philip.brown@uccs.edu}}}
}

\begin{document}
\maketitle
\begin{abstract}
Are rewards or penalties more effective in influencing user behavior?
This work compares the effectiveness of subsidies and tolls in incentivizing users in congestion games.
The predominantly studied method of influencing user behavior in network routing problems is to institute taxes which alter users' observed costs in a manner that causes their self-interested choices to more closely align with a system-level objective.
Another feasible method to accomplish the same goal is to subsidize the users' actions that are preferable from a system-level perspective.
We show that, when users behave similarly and predictably, subsidies offer comparable performance guarantees to tolls while requiring smaller monetary transactions with users; however, in the presence of unknown player heterogeneity, subsidies fail to offer the same performance as tolls.
We further investigate these relationships in affine congestion games, deriving explicit performance bounds under optimal tolls and subsidies with and without user heterogeneity; we show that the differences in performance can be significant.
\end{abstract}

\section{Introduction}
In systems governed by a collective of multiple decision making users, the system performance is often dictated by the choices those users make.
Though each user may make decisions rationally, the emergent behavior observed in the system need not align with the objective of the system designer.
This phenomenon appears in many highly studied settings including distributed control~\cite{shamma2007cooperative}, resource allocation problems~\cite{Johari2010,Paccagnan2018a}, and transportation networks~\cite{salazar2019congestion}.
An effective metric to quantify this apparent inefficiency is the \emph{price of anarchy}, defined as the worst-case ratio between the social welfare from users making self interested decisions and the optimal social welfare~\cite{Papadimitriou2001}.

A promising method of mitigating this inefficiency is by introducing incentives to the system's users, influencing their decisions to more closely align with the system optimal~\cite{Ratliff2018}.
One example of such incentives is to assign \emph{taxes}, eliciting monetary fees from users that will affect their preferences over the available actions~\cite{Ferguson2019,Cole2003,Fotakis2010}.
Such taxes have been shown to be effective in reducing system inefficiency with respect to the price of anarchy ratio~\cite{Bilo2016,Caragiannis2010,Chandan2019computing,Fleischer2004}.
Another method to influence user behavior is to \emph{subsidize} actions that are preferential from a system designer perspective.
Though the use of subsidies is equally feasible in theory and in implementation (less any budgetary constraints on the system designer), this method has been studied significantly less than the tax equivalent; the relative performance of each is thus unknown.

In this paper, we seek to understand the relative performance of subsidies and taxes in influencing user behavior in sociotechnical systems such as transportation networks.

Specifically, we consider a network routing problem in which users must traverse a network with congestable edges with delays that grow as a function of the local mass of users.
Finding a route for each user that minimizes the total latency in the system is straightforward if the system designer has full control in directing the users.
However, when users select their own routes, the resulting network flow need not be optimal~\cite{Pigou1920}.
Modeling the selfish routing problem as a non-atomic congestion game, we adopt the \emph{Nash flow} as a solution concept of the emergent behavior in the system.
Due to users selfish routing, the price of anarchy may be large~\cite{Roughgarden2005}.
It is for this reason that we introduce incentives which alter the players' observed costs so that the resulting Nash flow will more closely resemble the optimal flow.

A well studied method of incentivizing users in congestion games is to tax the users, or introduce tolls to paths in the network~\cite{Ferguson2019,Cole2003,Fotakis2010,Bilo2016,Caragiannis2010,Chandan2019computing,Fleischer2004,Bonifaci2011,Fotakis2007}.
In each of these settings, the effectiveness of these tolls is measured via the price of anarchy.
Indeed in the most elementary settings, tolls exist that influence users to self route inline with the system optimum~\cite{Pigou1920}.
However, when more nuance is introduced in the form of \emph{player heterogeneity} (i.e., players differing in their response to incentives), the task of designing tolls becomes more involved.
If the toll designer possesses sufficient knowledge of the network structure and user population, they may still compute and implement tolls which incentivize optimal routing~\cite{Fleischer2004};
however, in the case where the system designer has some uncertainty in the network parameters or behavior of the user population, tolls are often designed to minimize the price of anarchy ratio~\cite{Fotakis2010,Karakostas2004,Brown2017d}, and again, encouraging results exist.
Many of these works however, do not consider the magnitude of the tax levied on the users.
To this end, we consider that additional budgetary constraints may be imposed in the form of \emph{bounds} on the incentives.
The authors of~\cite{Bonifaci2011,Brown2017d} consider such restrictions on tolls, but only in the context of small classes of networks and restrict incentives to positive taxes.

Though the study of tolling in congestion games is extensive, there are few results regarding subsidies as incentives in this context.
In~\cite{Arieli2015}, the authors investigate budget-balancing tolls in which the amount of tax incentives equals the amount of subsidies, but the results exist only with homogeneous users.
The authors of \cite{Sandholm2002} consider more general incentives, but in an evolutionary setting.
From a system designer's perspective, subsidies may be a feasible method of influencing user behavior; the performance guarantees of subsidies is thus of interest as well as how this performance compares to tax incentives.

In this work, we compare the performance of tolls and subsidies in congestion games in the presence of budgetary constraints and user heterogeneity.
Our results are outlined as follows:

\begin{hangparas}{.25in}{1}
\textbf{\cref{thm:bounded}:} With homogeneous users, we show that, with a similar budgetary constraint, the optimal bounded subsidy performs no worse than the optimal bounded toll and in many cases performs strictly better.
\end{hangparas}
\begin{hangparas}{.25in}{1}
\textbf{\cref{thm:robust}:} In contrast, when users differ in their response to incentives, we give a result that shows tolls are inherently more robust to player heterogeneity than subsidies.
\end{hangparas}
\begin{hangparas}{.25in}{1}
\textbf{Proposition~\ref{prop:aff_bounded}:} We characterize the optimal bounded tolls and subsidies in affine congestion games along with their associated performance guarantees.
\end{hangparas}
\begin{hangparas}{.25in}{1}
\textbf{Proposition 4:} We identify the effect of player heterogeneity on similar toll and subsidy mechanisms in affine congestion games.
\end{hangparas}

The first two results hold for general classes of non-atomic congestion games, and illustrate the relation between subsidies and tolls: for homogeneous users, subsidies can offer similar performance to tolls with smaller monetary transactions, but the presence of player heterogeneity degrades this performance more dramatically with subsidies than with tolls.
The final propositions illustrate this result in the well studied class of affine congestion games

\section{Preliminaries}
\subsection{System Model}
Consider a directed graph $(V,E)$ with vertex set $V$, edge set $E \subseteq (V \times V)$, and $k$ origin-destination pairs $(o_i,d_i)$.
Denote by $\paths_i$ the set of all simple paths connecting origin $o_i$ to destination $d_i$.
Further, let $\paths = \cup_{i=1}^k \paths_i$ denote the set of all paths in the graph.
A \textit{flow} on the graph is a vector $f \in \mathbb{R}^{|\paths|}_{\geq 0}$ that expresses the mass of traffic utilizing each path.
The mass of traffic on an edge $e \in E$ is thus $f_e = \sum_{P:e\in \paths} f_P$, and we say $f = \{f_e\}_{e \in E}$.
A flow $f$ is \textit{feasible} if it satisfies $\sum_{P \in \paths_i}f_P = r_i$ for each source-destination pair, where $r_i$ is the mass of traffic traveling from origin $o_i$ to destination $d_i$.

Each edge $e \in E$ in the network is endowed with a non-negative, non-decreasing \textit{latency} function $\ell_e:\mathbb{R}_{\geq 0} \rightarrow \mathbb{R}_{\geq 0}$ that maps the mass of traffic on an edge to the delay users on that edge observe.
The system cost of a flow $f$ is the \textit{total latency},
\be
	\mathcal{L}(f) = \sum_{e \in E} f_e \cdot \ell_e(f_e).
\ee
A \textit{routing problem} is specified by the tuple $G = (V,E,\{\ell_e\}_{e\in E},\{r_i,(o_i,d_i)\}_{i=1}^k)$, and we  let $\mathcal{F}(G)$ denote the set of all feasible flows.
We define the optimal flow $\fopt$ as one that minimizes the total latency, i.e.,
\be
	\fopt \in \argmin_{f \in \mathcal{F}(G)} \mathcal{L}(f).
\ee
We will denote a family of routing problems by $\gee$.

\subsection{Incentives \& Heterogeneity}
In this paper, we consider the problem of selfish routing, where each user in the system chooses a path as to minimize their own observed delay.
Define $N_i$ as the set of users traveling from origin $o_i$ to destination $d_i$.
Each user $x \in N_i$ is thus free to choose between paths $P \in \paths_i$.
Let each $N_i$ be a closed interval with Lebesgue measure $\mu(N_i) = r_i$ that is disjoint from each other set of users, i.e., $N_i \cap N_j = \emptyset \ \forall \ i,j \in \{ 1,\ldots,k \}, \ i \neq j$.
The full set of agents is thus $N = \cup_{i=1}^k N_i$ whose mass is $\mu(N) = \sum_{i=1}^k r_i$.
Without loss of generality, we can let $\mu(N) = 1$.

It is well known that selfish routing can lead to sub-optimal system performance~\cite{Roughgarden2005}.
It is therefore up to a system designer to select a set of \textit{incentive functions} $\tau_e:[0,1] \rightarrow \mathbb{R} \ \forall e\in E$ to influence the behavior of the users in the system to more closely align with the system optimal flow.
These incentives can be regarded as monetary transfers with the users dependent on the paths they choose.

In this work, we consider that users may differ in their response to incentives.
Specifically, each player $x \in N$ is associated with a sensitivity $s_x \geq 0$ to incentives.
We call $s:N \rightarrow \mathbb{R}_{\geq 0}$ a \textit{sensitivity distribution}. 
We highlight the case where $s_x = c \ \forall x\in N$ for some constant $c$ as a \textit{homogeneous} distribution of user sensitivities\footnote{Without loss of generality, we use $s_x=1$ for a homogeneous population.}, in which each user behaves similarly; any other distribution will be referred to as a population of \textit{heterogeneous} users.

A user $x \in N_i$ traveling on a path $P_x \in \paths_i$ will observe cost
\be \label{eq:player_cost}
	J_x(P_x,f) = \sum_{e \in P_x} \ell_e(f_e) + s_x\tau_e(f_e).
\ee
A flow $f$ is a \textit{Nash flow} if
\begin{multline}
	J_x(P_x,f) \in \argmin_{P \in \paths_i} \left\lbrace \sum_{e\in P} \ell_e(f_e) + s_x\tau_e(f_e) \right\rbrace \\ \forall x\in N_i, ~ i \in \{1,\ldots,k\}.
\end{multline}

A game is therefore characterized by a network $G$, player sensitivity distribution $s$, and a set of incentive functions $\{\tau_e\}_{e\in E}$, denoted by the tuple $(G,s,\{\tau_e\}_{e\in E})$. It is shown in \cite{Mas-Colell1984} that a Nash flow will exist in a congestion game of this form if the latency and incentive functions are Lebesgue-integrable.

\subsection{Incentive Mechanisms \& Performance Metrics}
To determine the manner in which incentive functions are applied to edges, we investigate \textit{incentive mechanisms}.
To formalize this notion, we denote $L(G)$ as the set of latency functions in the routing problem $G$.
Further, for a family of problems, we denote $L(\gee) = \cup_{G\in \gee}L(G)$ as the set of latency functions that occur in the family of games $\gee$.

An example of this would be the well studied family of affine congestion games $\gee^{\rm aff}$, in which each edge has an affine latency function, i.e., $$L(\gee^{\rm aff}) = \{\ell(f)=af+b \ | \ a\geq 0, \ b\geq 0\}.$$
An incentive mechanism $T$ assings to an edge $e$ with latency $\ell_e$ an incentive $T(\ell_e)$, i.e. $\tau_e(f_e) = T(\ell_e)[f_e]$, where $T(\ell_e)[f_e]$ is the incentive evaluated at $f_e$.
This mapping is denoted $T:L(\gee) \rightarrow \Tau$ where $\Tau$ is some set of allowable incentive functions.
Using this framework allows us to consider a large variety of incentives and congestion games, including cases where $\gee$ possess a finite or infinite number of games, subsuming network-aware~\cite{Fleischer2004,Bilo2016} and network-agnostic~\cite{Brown2017d,Chandan2019computing} methods.
One such incentive that fits this framework is the classic Pigouvian or marginal cost tax,
\be \label{eq:marginal_cost}
	T^{\rm mc}(\ell)[f] = f\cdot \frac{d}{df} \ell(f),
\ee
which is known to incentivize users to route optimally in many classes of congestion games~\cite{Pigou1920}, i.e., $\fnash = \fopt$.
This is only true however, when there is no bound on the incentive and users are homogeneous~\cite{Brown2017d}.

To quantify the robustness of an incentive mechanism, we also consider that the system designer may be unaware of users' response to incentives.
We denote a set of sensitivity distributions $\sdist = \{s:N \rightarrow [\bmin,\bmax] \}$, where $\bmin>0$ is a lower bound on users' sensitivity to incentives and $\bmax\geq \bmin$ is an upper bound; we include these bounds to quantify the range of users responses, signifying the amount of possible player heterogeneity.

Let $\Lnash(G,s,T)$ be the highest total latency in a Nash flow of the game $(G,s,T(L(G)))$.
Additionally, let $\Lopt(G)$ be the total latency under the optimal flow $\fopt$.
We use the \textit{price of anarchy} to evaluate the performance of a taxation mechanism, defined as the worst case ratio between total latency in a Nash flow and an optimal flow.
This inefficiency can be characterized by
\vs
\be
	\poa(G,s,T) = \frac{\Lnash(G,s,T)}{\Lopt(G)}.
\ee
We extend this definition to a family of instances
\vs
\be
	\poa(\gee,\sdist,T) = \sup_{G \in \gee} \ \sup_{s \in \sdist} \ \frac{\Lnash(G,s,T)}{\Lopt(G)},
\ee
where $T$ is used in each network, sensitivity distribution pair.
The price of anarchy is now the worst case inefficiency over all such network, sensitivity distribution pairs while using incentive mechanism $T$.
In the case of homogeneous sensitivity distributions, we simply omit $s$ and write $\poa(\gee,T)$.
The objective of such incentive mechanisms is to minimize this worst case inefficiency, thus the optimal incentive mechanism will be defined as,
\vs
\be
	T^{\rm opt} \in \underset{T: L(\gee) \rightarrow \Tau}{\arginf} \poa(\gee,\sdist,T),
\ee
such that it minimizes the price of anarchy for a class of games $\gee$ and set of sensitivity distributions $\sdist$.

\subsection{Tolls \& Subsidies}
We differentiate between two forms of incentives, \textit{tolls} $\tau_e^+:[0,1]\rightarrow \mathbb{R}_{\geq 0}$ and \textit{subsidies} $\tau_e^-:[0,1]\rightarrow \mathbb{R}_{\leq 0}$.
With tolls, the player's observed cost is strictly increased, i.e., the system designer levies taxes for the users to pay depending on their choice of edges.
With subsidies, the players cost is strictly reduced, i.e., the system designer offers some payments to users for their choice of action.
The main focus of this work is to assess which is more effective in influencing user behavior, tolls or subsidies.
To explore this, we further introduce bounded tolls and subsidies.
A bounded toll will satisfy $\tau_e^+(f_e) \in [0,\beta\cdot \ell_e(f_e)]$ for $f_e \in [0,1]$ and each $e \in E$, where $\beta$ is a bounding factor.
A bounded tolling mechanism will be denoted $T^+(\ell_e;\beta)$.
Similarly, a bounded subsidy will satisfy $\tau_e^-(f_e) \in [-\beta\cdot \ell_e(f_e),0]$ for $f_e \in [0,1]$ and each $e \in E$, and a bounded subsidy mechanism will be denoted $T^-(\ell_e;\beta)$.

To compare the efficacy of bounded tolls and subsidies, we define an optimal bounded tolling mechanism as 
\begin{equation}
	T^\mathrm{opt+}(\beta,\sdist) \in \underset{T^+:\mathcal{L}\rightarrow \mathcal{T}^+(\beta)}{\arginf} \ \poa(\gee,\sdist,T^+),
	\vs \vs \vs
\end{equation}
where $$\mathcal{T}^+(\beta) = \{\tau_e^+ \in \Tau \ | \ \tau_e^+(f_e) \in [0,\beta \cdot \ell_e(f_e)] ~ \forall f_e\in [0,1] \}$$ is the set of all tolling functions.
The optimal bounded subsidy mechanism $T^{\rm opt-}(\beta,\sdist)$ is defined analogously.
For notational convenience, we will omit the dependence on $\sdist$ in the homogeneous setting.

\subsection{Summary of Our Contributions}
In this setting, questions that arise are
\begin{enumerate}
\item How do budgetary constraints on incentives affect the price of anarchy for homogeneous users? 
\item How does player heterogeneity affect the price of anarchy of tolls and subsidies?
\end{enumerate} 
We address the first question by showing in \cref{thm:bounded} that, with homogeneous users, the price of anarchy of the optimal subsidy $T^{\rm opt-}(\beta)$ is no greater than the price of anarchy of the optimal toll $T^{\rm opt+}(\beta)$ under the same bounding factor; we further show this relation is strict in many non-trivial classes of congestion games.
In \cref{thm:robust}, we address question two by showing that for a class of tolls and a class of subsidies with the same price of anarchy guarantee in the homogeneous setting, when player heterogeneity is introduced the optimal robust toll has lower price of anarchy than the optimal robust subsidy.
We offer closed for price of anarchy bounds in Proposition~\ref{prop:aff_bounded} and Proposition~\ref{prop:aff_robust} in the case of affine congestion games to illustrate these result in a well studied setting.

\section{General Results}
\subsection{Bounded Incentives}
Though we consider any toll bound $\beta \geq 0$, we offer the following definition to differentiate from cases where the bound is insignificantly large or trivially zero.

\begin{definition}
A toll (subsidy) is \emph{tightly bounded} if \mbox{$\tau(f)=\beta\ell(f)$,} (if $\tau(f)=-\beta\ell(f)$) for some $f \in (0,1]$.
\end{definition}

When an optimal incentive is tightly bounded, the budgetary constraint is active.
With this in mind, \cref{thm:bounded} states that bounded subsidies will outperform similarly bounded tolls with respect to the price of anarchy, and strictly outperform when the budgetary constraint is active.

\begin{theorem}\label{thm:bounded}
For a family of congestion games $\gee$, under a bounding factor $\beta \geq 0$ the optimal subsidy mechanism $T^{\rm opt-}(\beta)$ will have no greater price of anarchy than the optimal tolling mechanism $T^{\rm opt+}(\beta)$, i.e.,
\begin{equation}\label{eq:thmPoAincentive}
 	\poa \left(\gee,T^{\rm opt+}(\beta) \right) \geq \poa \left(\gee,T^{\rm opt-}(\beta) \right) \geq 1.
\end{equation}
Additionally, if every optimal subsidy is tightly bounded, then the first inequality in \eqref{eq:thmPoAincentive} is strict.
\end{theorem}
The proof of \cref{thm:bounded} appears at the end of this subsection; we first discuss the implications of this result.
\cref{thm:bounded} implies that when limiting the size of monetary transactions with users, subsidies are more effective than tolls at influencing user behavior.
This result holds for any class of non-atomic congestion games, including those containing finite number of games.
Though \eqref{eq:thmPoAincentive} need not be strict in general, there does exist a gap between the performance of tolls and subsidies in many, non-trivial settings.


To illustrate this, we offer the following example to highlight that bounded subsidies outperform bounded tolls in a well studied class of congestion games.
\begin{example}
\emph{Polynomial Congestion Games.} Consider the class of congestion games $\gee^p$, in which, 
$$L(\gee^p) = \left\lbrace \ell(f) = \sum_{i=0}^p \alpha_i f^p \ | \ \alpha_i \geq 0 \ \forall \ i \in \{0,\ldots,p \} \right\rbrace,$$
typically referred to as polynomial congestion games.
Specifically, consider a game with graph $G$, possessing two nodes forming a source destination pair with unit mass of traffic and two parallel edges between them.
This example, depicted in \cref{fig:polypigou}, has been shown to demonstrate the worst case inefficiency among polynomial congestion games~\cite{Roughgarden2003}.

When users are homogeneous in their sensitivity to incentives, an optimal toll for this class of games is the marginal cost toll in \eqref{eq:marginal_cost}, proven to incentivize optimal routing~\cite{Pigou1920}.
Notice that this marginal-cost toll will manifest in this network as 
\vs \vs \vs
\begin{eqnarray}
\tau^{\rm mc}_1(f_1) = pf_1^p, & \tau^{\rm mc}_2(f_2) = 0,
\end{eqnarray}
and indeed incentivize the Nash flow to be the system optimal of $f_1 = 1/\sqrt[\leftroot{-3}\uproot{3}p]{p+1}$.
\begin{observation}\label{obs:optimal}
The set of incentive mechanisms
\vs
$$\left\lbrace T(\ell)= \lambda T^{\rm mc}(\ell) + (\lambda-1)\ell \ | \ \lambda >0 \right\rbrace,$$
\vs
contains all taxation mechanisms that guarantee a price of anarchy of 1.
\end{observation}
This observation can be proven from methods similar to Theorem 4.1 in~\cite{Brown2017b} and \cref{lem:scaling}.
Intuitively, any other incentive mechanism will cause users to observe different costs on an edge, altering the Nash flow from optimal in some game $G \in \gee^p$, offering a lower-bound above one.

Now, suppose $\beta \in [\frac{p}{p+1},p)$, the marginal-cost toll is not feasible, nor is any optimal toll from \cref{obs:optimal}.
However, consider the subsidy 
$$T^-(\ell) = (1/(p+1) - 1)\ell + (1/(p+1))T^{\rm mc}(\ell),$$
 which exists in the set of optimal incentives.
This subsidy will manifest in the network in \cref{fig:polypigou} as
\vs \vs
\begin{eqnarray}
\tau^{\rm mc}_1(f_1) = 0, & \tau^{\rm mc}_2(f_2) = \frac{-p}{p+1},
\end{eqnarray}
and incentivize optimal routing.
Because $\beta \in [\frac{p}{p+1},p)$, this subsidy is still feasible under a similar budgetary constraint.
Thus, a bounded subsidy exists that gives a price of anarchy of one while a similar bounded toll does not exist, giving the strict difference between subsidies and tolls from~\cref{thm:bounded}.

For $\beta \in (0,\frac{p}{p+1})$, though the optimal bounded tolls and subsidies remains an open question, we conjecture they will be tightly bounded by the budgetary constraint; thus by \cref{thm:bounded}, the optimal bounded subsidy will continue to strictly outperform the optimal bounded toll.
Under tightly bounded incentives, we can see in \cref{fig:polypigou} the magnitude of difference between tolls and subsidies.
\end{example}


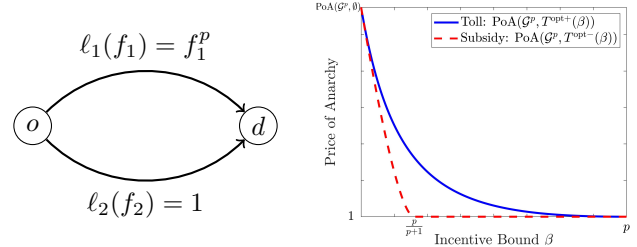
\begin{figure}[t!]
\vspace{2mm}
    \centering
    \begin{subfigure}[t!]{0.235\textwidth}\centering
    \centering
    \begin{tikzpicture}
        \node () [align=center, above] at (1.5,0.75) {$\ell_1(f_1)=f_1^p$};
        \node () [align=center, below] at (1.5,-0.75) {$\ell_2(f_2)= 1$};
        \node (o) [circle,draw,inner sep=0pt,fill=white,minimum width=0.5cm]  at (0.0,0) {$o$};
        \node (d) [circle,draw,inner sep=0pt,fill=white,minimum width=0.5cm]  at (3.0,0.0) {$d$};
        \draw [->,thick,out=45,in=135] (o) edge (d);
        \draw [->,thick,out=-45,in=-135] (o) edge (d);
    \end{tikzpicture}
    \label{fig:polypigou_graph}
    \end{subfigure}
    \begin{subfigure}[t!]{0.235\textwidth}\centering
        \includegraphics[width=\textwidth]{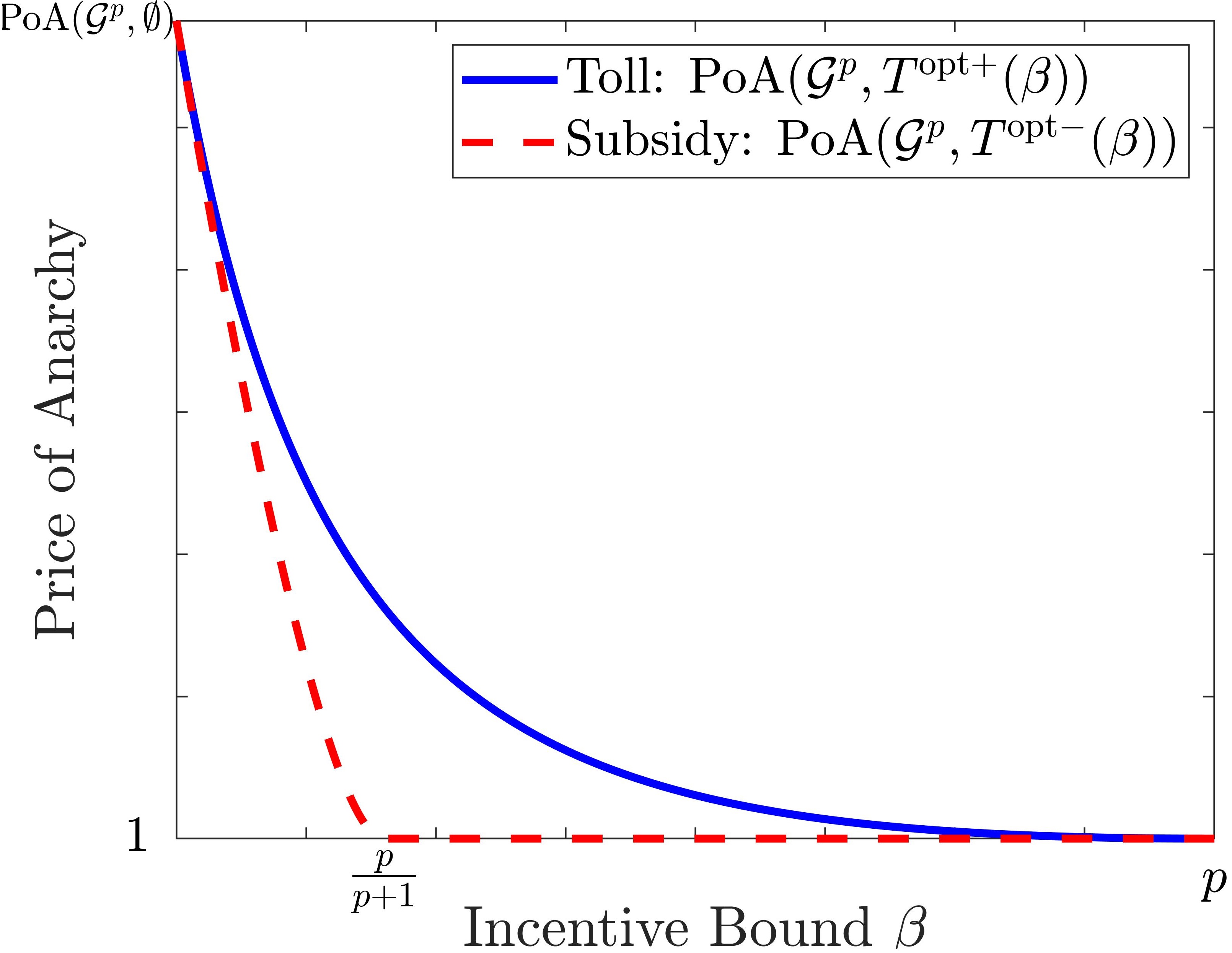}
        \label{fig:polypigou_poa}
        \vs \vs \vs \vs
    \end{subfigure}
    \caption{\small Price of anarchy in polynomial congestion games. (Left) worst case game in set of polynomial congestion games~\cite{Roughgarden2003}. (Right) Price of anarchy under bounded tolls and subsidies; for illustration, games with polynomial cost functions of degree no more than degree 4 are used.}
    \label{fig:polypigou}
    \vs \vs \vs \vs \vs \vs
\end{figure}

In Proposition~\ref{prop:aff_bounded}, we explicitly give the price of anarchy bounds of optimal bounded tolls and subsidies in affine congestion games, again demonstrating the strictly superior performance of subsidies as well as illustrating the magnitude of this difference in performance.

\noindent \textit{Proof of \cref{thm:bounded}:}

We first show a transformation on incentive mechanisms that does not affect the price of anarchy under homogeneous user sensitivities.
This transformation gives us the important relationship between incentive mechanisms that their performance is not unique and similar performance can be garnered with different magnitudes of transactions.
\begin{lemma}\label{lem:scaling}
Let $T:L(\gee)\rightarrow \mathcal{T}$ be an incentive mechanism over the family of congestion games $\gee$. If another influencing mechanism is defined as $\Tlam(\ell_e) = \lambda T(\ell_e) + (\lambda - 1)\ell_e$ for some $\lambda > 0$, then
\vs \vs
\begin{equation}\label{eq:lem_scale}
\poa(\gee,T) = \poa(\gee,\Tlam).
\end{equation}
\end{lemma}
The proof of \cref{lem:scaling} appears in the appendix.

First, observe that if $\beta=0$ the only permissible influencing function for tolls and subsidies is $\tau_e^+(f_e) = \tau_e^-(f_e) = 0$, i.e., there is no influencing.
Therefore, the left and right hand side of \eqref{eq:thmPoAincentive} equate to the unincentivized case and the expression holds with equality.

Let $j_e(f_e) = \ell_e(f_e)+\tau_e(f_e)$ denote the cost a player observes for utilizing an edge $e$ when a mass of $f_e$ users are utilizing it.
The observed cost of a player $x \in N$ can be rewritten as $J_x(P_x,f) = \sum_{e\in P_x} j_e(f_e)$.
In the case where $\beta > 0$, a bounded tolling function on an edge must exist between $\tau_e^+(f_e) \in [0, \beta \cdot \ell(f_e)]$, and the edges observed cost will satisfy $j_e^+(f_e) \in [\ell_e(f_e), (1+\beta)\cdot \ell_e(f_e)]$.
Similarly, a subsidy function on an edge must exist between $\tau_e^-(f_e) \in [-\beta \cdot \ell(f_e),0]$, and the edges observed cost will satisfy $j_e^-(f_e) \in [(1-\beta)\cdot \ell_e(f_e), \ell_e(f_e)]$.

Let $T^+(\ell_e;\beta)$ be a bounded tolling mechanism with edge costs of $j^+_e(f_e)$.
Now, define $\Tlam(\ell_e) = \lambda T^+(\ell_e;\beta) + (\lambda-1) \ell_e$; from \cref{lem:scaling}, $T^+$ and $\Tlam$ will have the same price of anarchy for any $\lambda > 0$.
Let $\hat{j}_e$ be the edge cost under influencing mechanism $\Tlam$, from the construction of $\Tlam$
\begin{equation}\label{eq:ftransfer}
\hat{j}_e = \ell_e + \Tlam(\ell_e) = \ell_e + \lambda T^+(\ell_e;\beta) + (\lambda-1) \ell_e = \lambda j_e^+.
\end{equation}

We now look at the cases where $\beta \in (0,1)$ and $\beta \geq 1$ respectively.
When $\beta \in (0,1)$, let $\lambda = (1-\beta)$.
Now, 
\begin{align*}
\hat{j}_e(f_e) = (1-\beta)j_e^+(f_e) &\in [(1-\beta)\ell_e(f_e), (1-\beta^2)\ell_e(f_e)]\\
& \subset [(1-\beta) \ell_e(f_e), \ell_e(f_e)],
\end{align*}
thus the edge costs are sufficiently bounded such that $\Tlam$ is a permissible subsidy mechanism bounded by $\beta$ with the same price of anarchy as $T^+$.
If $\beta \geq 1$ let $\lambda = 1/(1+\beta)$ and get
\begin{align*}
\hat{j}_e(f_e) = \frac{1}{(1+\beta)}j_e^+(f_e) &\in \left[\frac{1}{(1+\beta)}\ell_e(f_e), \ell_e(f_e)\right] \\
& \subset [(1-\beta) \ell_e(f_e), \ell_e(f_e)],
\end{align*}
and again $\Tlam$ is a permissible subsidy mechanism bounded by $\beta$.
By letting $T^+ = T^{\rm opt+}$ we obtain \eqref{eq:thmPoAincentive}.

We have proven that, for $\beta > 0$, if $\poa(\gee,T^{\rm opt-}(\beta)) = \poa(\gee,T^{\rm opt+}(\beta))$, then there exists a $T^{\rm opt-}(\beta)$ that does not achieve the bound.
The contrapositive of this is that if every optimal subsidy achieves the bound, the price of anarchy guarantees are not equal.
In this case, the optimal subsidies are each tightly bounded and $\poa(\gee,T^{\rm opt-}(\beta)) < \poa(\gee,T^{\rm opt+}(\beta))$, proving the final part of~\cref{thm:bounded}.
$\qed$

\subsection{Robustness to User Heterogeneity}
We now look at the performance of tolls and subsidies when users need not have identical responses to incentives.
Often, increased player heterogeneity causes the performance of an incentive mechanism to diminish.
We give the following definition for the classes of congestion games with this property\footnote{There do exist examples of games that do not satisfy Definition~\ref{def:nonresist}; these classes of games are often trivial, including those with a single path connecting an origin destination pair.}.
\begin{definition}\label{def:nonresist}
A class of congestion games is \emph{responsive to player heterogeneity} if $\poa(\gee,\sdist,T)$ is strictly increasing with $\bmax/\bmin > 1$ for any incentive mechanism $T$.
\end{definition}
These classes of games are those that have a degradation in performance from increased player heterogeneity; many classes of well studied congestion games possess this property~\cite{Brown2017d}.

In \cref{thm:robust}, we give a robustness result that illustrates how player heterogeneity has a more drastic impact on the performance of subsidies than of tolls.

\begin{theorem}\label{thm:robust}
For a class of congestion games $\gee$, define two incentive bounds $\beta^+$ and $\beta^-$ such that
\vs
\begin{equation}
	 \poa \left(\gee,T^{\rm opt-}(\beta^-)) = \poa(\gee,T^{\rm opt+}(\beta^+) \right),
\end{equation}
\vs
then at the introduction of player heterogeneity,
\begin{equation}\label{eq:thm_rob_uneq}
	\poa \left(\gee,\sdist,T^{\rm opt-}(\beta^-,\sdist)) \geq \poa(\gee,\sdist,T^{\rm opt+}(\beta^+,\sdist) \right) \geq 1.
\end{equation}
Additionally, \eqref{eq:thm_rob_uneq} is strict if $\gee$ is responsive to player heterogeneity.
\end{theorem}

%
%

Intuitively, this result stems from the fact that subsidies are more finely tuned to give performance guarantees.
This fact causes the same amount of player heterogeneity to have a larger effect on Nash flows than with an equivalent toll.
Thus, when increased player heterogeneity escalates the inefficiency, this relationship is strict.
Though the relationship isn't strict for general classes of congestion games, it is for many well studied cases, including the aforementioned polynomial congestion games.
In Proposition~\ref{prop:aff_robust}, we give price of anarchy bounds for robust incentives in affine congestion games and deduce the magnitude of the difference in performance.

\noindent \textit{Proof of \cref{thm:robust}:}

First, we give the following definition for incentives that have the same performance in the homogeneous setting.
\begin{definition}
For some incentive mechanism $T$, each incentive mechanism satisfying $\Tlam(\ell_e) = (\lambda-1)\ell_e + \lambda T(\ell_e)$ for some $\lambda > 0$ is termed \textit{nominally equivalent}.
From~\cref{lem:scaling}, nominally equivalent incentives will satisfy, 
\be
	\poa(\gee,T) = \poa(\gee,\Tlam).
\ee
\end{definition}

We show, in \cref{lem:hetero}, a relation between nominally equivalent incentives in the heterogeneous population setting; specifically, we show that the heterogeneous price of anarchy decreases as incentives increase costs to the users.

\begin{lemma}\label{lem:hetero}
For a class of congestion games $\gee$, let $T$ be an incentive mechanism.
If $\Tlam$ is nominally equivalent to $T$, then $\poa(\gee,\sdist,\Tlam)$ is monotonically decreasing with $\lambda$.
\end{lemma}
The proof of \cref{lem:hetero} appears in the appendix.

The theorem follows closely from \cref{lem:scaling} and \cref{lem:hetero}.
It can be inferred from \ref{lem:scaling} that, for the given bounds $\beta^+$ and $\beta^-$, there exists an optimal toll $T^{\rm opt+}(\beta^+)$ and an optimal subsidy $T^{\rm opt-}(\beta^-)$ that are nominally equivalent.
In fact, any toll bounded by $\beta^+$ will have a nominally equivalent subsidy bounded by $\beta^-$.

Now, let $T^{\rm opt-}(\beta^-,\sdist)$ be the optimal subsidy with player heterogeneity bounded by $\beta^-$.
From the fact before, we know there exists a toll $T^+$ that is nominally equivalent to $T^{\rm opt-}(\beta^-,\sdist)$ and bounded by $\beta^+$.
From \cref{lem:hetero}, we obtain that
\vs
\begin{equation}\label{eq:rob_minus}
	\poa(\gee,\sdist,T^+) \leq \poa(\gee,\sdist,T^{\rm opt-}(\beta^-,\sdist)),
\end{equation}
and by the definition of $T^{\rm opt+}(\beta^+,\sdist)$, we get
\vs
\begin{equation}\label{eq:rob_plus}
	\poa(\gee,\sdist,T^{\rm opt+}(\beta^+,\sdist)) \leq \poa(\gee,\sdist,T^+).
\end{equation}
Combining \eqref{eq:rob_minus} and \eqref{eq:rob_plus} gives \eqref{eq:thm_rob_uneq}.
If the class of games is responsive to player heterogeneity, then $\poa(\gee,\sdist,\Tlam)$ is strictly decreasing with $\lambda$ and the relationship is strict.
$\qed$

\section{Affine Congestion Games}
As a means of illustrating some of the results above, we look at the well studied class of affine congestion games. Here, we give results on optimal bounded tolls and subsidies and their accompanying performance guarantees whose relation follow the general results.
We include these results to highlight the appreciable gap in performance between subsidies and tolls in the homogeneous and heterogeneous cases respectively.

\begin{figure}[t!]
\vspace{2mm}
    \centering
    \begin{subfigure}[t!]{0.235\textwidth}
        \includegraphics[width=\textwidth]{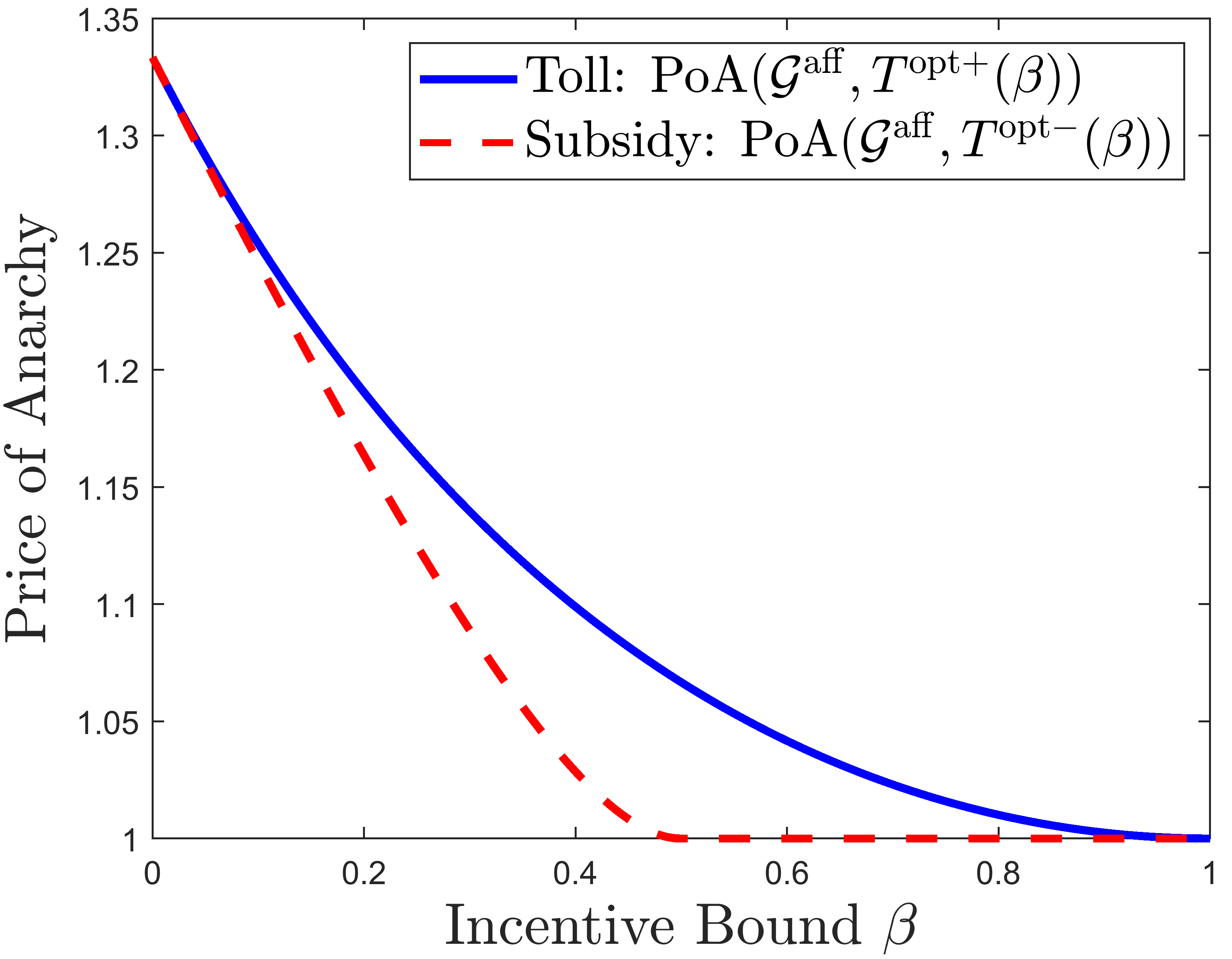}
        \caption{\raggedright \small Price of anarchy with bounded incentives}
        \label{fig:poa_bounded}
    \end{subfigure}
    \begin{subfigure}[t!]{0.235\textwidth}
        \includegraphics[width=\textwidth]{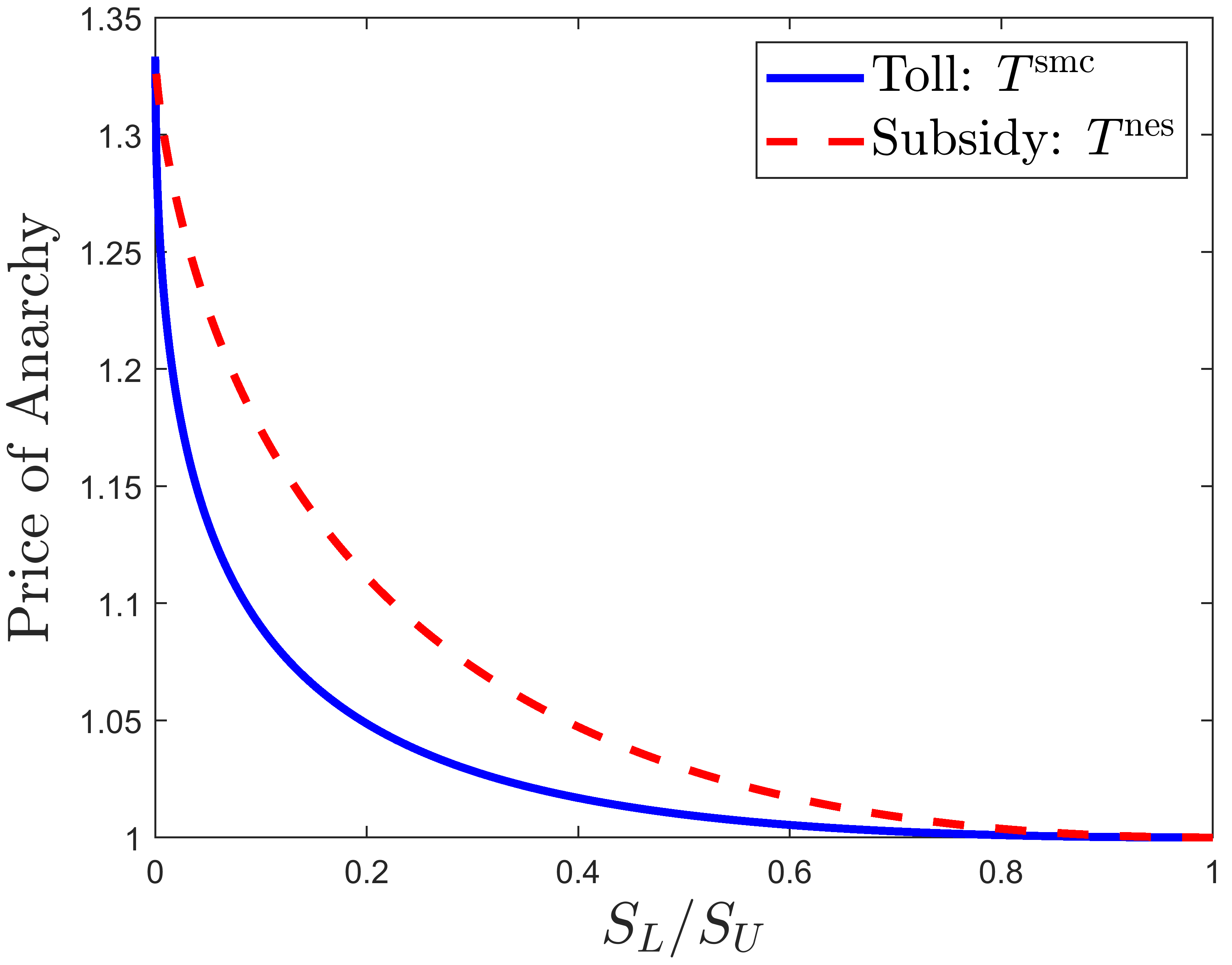}
        \caption{\raggedright \small Price of anarchy with player heterogeneity}
        \label{fig:poa_robust}
    \end{subfigure}
    \caption{\small Price of Anarchy bounds for comparable tolls and subsidies in affine congestion games. (Left) Price of Anarchy under optimal toll and subsidy respectively bounded by a factor $\beta$. (Right) Price of Anarchy of a nominally equivalent toll and subsidy with heterogeneity of user sensitivity introduced; $\bmin/\bmax$ expresses the amount of possible heterogeneity in the population.}\label{fig:poa}
    \vs \vs \vs \vs \vs \vs
\end{figure}

\subsection{Optimal Bounded Incentives in Affine Congestion Games}
We first state the optimal bounded tolls and optimal bounded subsidies in Proposition~\ref{prop:aff_bounded}, as well as give the associated price of anarchy guarantees.
The price of anarchy bound for the optimal bounded toll and subsidy can be seen for various incentive bounds in~\cref{fig:poa_bounded}.
Observe that the optimal subsidy outperforms the optimal toll for each incentive bound, matching the results from~\cref{thm:bounded}.

\begin{proposition}\label{prop:aff_bounded}
	The optimal bounded tolling mechanism in $\gee^{\rm aff}$ is
	\vs \vs
	\begin{equation}
	T^{\rm opt+}(af+b;\beta) = \begin{cases} 
				      \beta a x & \beta \in [0,1), \\
      				  a x & \beta \geq 1, 
   					\end{cases}
	\end{equation}
	with a price of anarchy bound of
	\begin{equation}
	\poa(\gee^{\rm aff},T^{\rm opt+}(\beta)) = \begin{cases} 
				      \frac{4}{3+2\beta-\beta^2} & \beta \in [0,1), \\
      				  1 & \beta \geq 1. 
   					\end{cases}
	\end{equation}
	Additionally, the optimal bounded subsidy mechanism in $\gee^{\rm aff}$ is
	\vs \vs
	\begin{equation}
	T^{\rm opt-}(af+b;\beta) = \begin{cases} 
				      -\beta b & \beta \in [0,1/2), \\
      				  -b/2 & \beta \geq 1/2, 
   					\end{cases}
	\end{equation}
	with a price of anarchy bound of
	\begin{equation}
	\poa(\gee^{\rm aff},T^{\rm opt-}(\beta)) = \begin{cases} 
				      \frac{4}{3+2\hat{\beta}-\hat{\beta}^2} & \beta \in [0,1/2), \\
      				  1 & \beta \geq 1/2, 
   					\end{cases}
	\end{equation}
	where $\hat{\beta} = 1/(1-\beta)-1$. Accordingly, for any $\beta \in (0,1)$,
	\be
		\poa(\gee^{\rm aff},T^{\rm opt+}(\beta)) < \poa(\gee^{\rm aff},T^{\rm opt-}(\beta)).
	\ee
\end{proposition}

\begin{proof}
We first look at the optimal bounded toll and its associated price of anarchy bound.
Trivially, when $\beta > 1$ the optimal toll is the marginal cost toll that gives price of anarchy of one.
For a bounding factor $\beta \in [0,1)$, a feasible bounded toll must satisfy
\be
	\tau_e^+(f_e) \in [0,\beta \cdot \ell_e] = [0,\beta a_ef_e+\beta b_e].
\ee
We can therefore write any feasibly bounded toll as $\tau_e^+(f_e) = k_1a_ef_e+k_2b_2$ where $k_1,k_2 \in [0,\beta]$.
We first show that the optimal toll will have $k_2=0$.

Let $T^+$ be a tolling mechanism that assigns bounded tolls with some $k_1,k_2 \in [0,\beta]$.
A player $x \in N_i$ utilizing path $P_x$ in a flow $f$ will observe cost
\be \label{eq:player_cost_k1k2}
	J_x(P_x,f) = \sum_{e\in P_x}(1+k_1)a_ef_e+(1+k_2)b_e.
\ee
Now, consider an incentive mechanism $\hat{T}$ where edges are assigned tolls $\tau_e(f_e) = (\frac{1+k_1}{1+k_2}-1)a_ef_e$.
Under this new incentive, the same player as before will now observe cost
\be \label{eq:player_cost_k1}
	\hat{J}_x(P_x,f) = \sum_{e\in P_x}\frac{1+k_1}{1+k_2}a_ef_e+b_e.
\ee
Because the player's cost in \eqref{eq:player_cost_k1k2} and \eqref{eq:player_cost_k1} are proportional, the players preserve the same preferences remain unchanged and the Nash flows remains unaltered.
Because $(\frac{1+k_1}{1+k_2}-1) \leq k_1 \leq \beta$ the new incentive will be bounded by $\beta$, and we need only consider tolls of the form $\tau_e(f_e) = ka_ef_e$ when in search of the optimal bounded toll.
When $k < 0$ the price of anarchy is at least $4/3$ and is indeed not optimal\footnote{Consider the classic Pigou network, as in \cref{fig:polypigou} with $p=1$. It is well known this network gives the worst case price of anarchy of $4/3$ with Nash flow of $f_1=1$. Consider using a taxation mechanism $T(af+b)=kaf$ for some $k<0$ and observe that the Nash flow is unchanged, thus not reducing the price of anarchy for the class of affine congestion games.}.

For a tolling mechanism $T^+(af+b) = kaf$ with $k \in [0,\beta) \subseteq [0,1)$, a player's cost will take the form
\be
	J_x(P_x,f) = \sum_{e\in P_x} (1+k) a_e + b_e.
\ee
When player cost functions take this form, the game is similar to that of an altruistic game (introduced in~\cite{Chen2014}) and has price of anarchy of
\vs
\be
	\poa(\gee^{\rm aff},T^+) = \frac{4}{3+2k-k^2}.
\ee
The price of anarchy is decreasing with $k \in [0,1)$ and thus the optimal toll occurs when $k$ is maximized at $k=\beta$.

For the optimal subsidy, we now note that incentives must be bounded by $\tau_e(f_e) \in [-\beta \ell_e(f_e),0]$.
From~\cref{lem:scaling}, we can map any such subsidy to an equivalent toll, now constrained to the region $\hat{\tau}_e(f_e) \in [0, \hat{\beta} \ell_e(f_e)]$ where $\hat{\beta} = (\frac{1}{1-\beta}-1)$.
It was shown prior that the optimal tolling mechanism in this region will be $\hat{T}(af+b) = \hat{\beta}af$.
Finally, we can again use~\cref{lem:scaling} to map back to the optimal bounded subsidy,
\be
	T^{\rm opt-}(af+b) = (\lambda-1)(af+b) + \lambda \hat{T}(af+b),
\ee
with $\lambda = 1-\beta$.
The result is an optimal subsidy of the form $T^{\rm opt-}(af+b)= -\beta b$ for $\beta \in [0, 1/2)$.
The price of anarchy bound comes from considering the equivalent toll.
\end{proof}

\subsection{Robustness of Incentives in Affine Congestion Games}
In this section, we specifically look at the optimal scaled marginal cost toll in parallel affine congestion games where each edge is utilized in the Nash flow, $\gee^{\rm pa}$, with player heterogeneity, $T^{\rm smc}(af+b) = (\sqrt{\bmin \bmax})^{-1}af$.
This tolling mechanism was first introduced in~\cite{Brown2017c}, and was shown to minimize the price of anarchy over parallel affine congestion games with sensitivity distributions in $\sdist$ bounded by $\bmin$ and $\bmax$.
In Proposition~\ref{prop:aff_robust}, we give price of anarchy bounds on the optimal scaled marginal cost toll as well as a nominally equivalent subsidy $T^{\rm nes}$.

\begin{proposition} \label{prop:aff_robust}
	Let $\gee^{\rm pa}$ be the set of fully-utilized parallel affine congestion games with sensitivity distributions in $\sdist$. The optimal scaled marginal cost tolling mechanism is \mbox{$T^{\rm smc}(af+b) = \frac{af}{\sqrt{\bmin \bmax}}$} with price of anarchy
	\vs \vs
	\be \label{eq:poa_smc}
		\poa(\gee^{\rm pa},\sdist,T^{\rm smc}) = \frac{4}{3} \left( 1-\frac{\sqrt{q}}{(1+\sqrt{q})^2} \right).
	\ee
	where $q:=\bmin/\bmax$. Additionally, a nominally equivalent subsidy will be $T^{\rm nes}(af+b) = -\frac{1}{1+\sqrt{\bmin \bmax}}b$, with price of anarchy
	\vs \vs
	\be \label{eq:poa_nes}
		\poa(\gee^{\rm pa},\sdist,T^{\rm nes}) = \frac{4}{3}\left( 1-\frac{\sqrt{\hat{q}}}{(1+\sqrt{\hat{q}})^2} \right),
	\ee
	where $$\hat{q} = \frac{\lambda q}{1-q+\lambda q} <q,$$
	and $\lambda = \sqrt{\bmin \bmax}/(1+\sqrt{\bmin \bmax})$.
\end{proposition}

Observe that, because $\hat{q}<q$ in \eqref{eq:poa_smc} and \eqref{eq:poa_nes} the nominally equivalent subsidy has greater price of anarchy when player heterogeneity is introduced.
This can be seen in \cref{fig:poa_robust}.
Intuitively, the same amount of player heterogeneity will have a larger effect on the subsidy than the toll.

\begin{proof}
This first part of the proposition comes from~\cite{Brown2017c}. We thus find the nominally equivalent subsidy mechanism and find the associated price of anarchy bound.

For notational convenience, let $k = 1/\sqrt{\bmin \bmax}$; the robust marginal cost toll is thus $T^{\rm smc}(af+b) = kaf$.
From \cref{lem:scaling}, we can derive a nominally equivalent subsidy by $T^{\rm nes}(af+b) = (\lambda - 1)(af+b) + \lambda (kaf)$, for and $\lambda > 0$.
By letting $\lambda = 1/(1+k)$, we get the nominally equivalent subsidy to be $T^{\rm nes}(af+b) = -kb/(1+k) = -\frac{1}{1+\sqrt{\bmin \bmax}}b$.

To determine the price of anarchy of $T^{\rm nes}$ with player heterogeneity, we use the result of \cref{thm:robust} to determine the equivalent level of heterogeneity on the nominally equivalent toll, $T^{\rm smc}$.
Let $s \in \sdist$ be a feasible sensitivity distribution, bounded by $\bmin$ and $\bmax$.
As it is defined above, we seek to find the preimage of $[\bmin,\bmax]$ under the function $g(S,1/(1+k))$.
Without loss of generality, we normalize $[\bmin,\bmax]$, to $[q,1]$ and look for it's preimage.
Because $g$ is continuous on $S \in [0,1]$, we look at the endpoints of the region.
We first note that $g(1,\lambda) = 1$ for any $\lambda > 0$.
Next, we determine $\hat{q}$ such that $g(\hat{q},\lambda) = q$ as
$$\hat{q} = \frac{\lambda q}{1-q+\lambda q},$$
and by setting $\lambda = 1/(1+k) = \sqrt{\bmin \bmax}/(1+\sqrt{\bmin \bmax})$ recover the equivalent amount of heterogeneity, $\hat{q}$, on $T^{\rm smc}$ as the original subsidy $T^{\rm nes}$ with heterogeneity $q$.
By replacing $q$ with $\hat{q}$ in \eqref{eq:poa_smc} we obtain the price of anarchy for $T^{\rm nes}$ with heterogeneity.
\end{proof}
\vs \vs 
\section{Conclusion}
This paper offers an extensive comparison of the performance of tolls and subsidies in congestion games.
In settings where budgetary constraints exist, subsidies appear to offer similar performance to tolls while requiring smaller monetary transactions with the users.
However, in the presence of player heterogeneity, it becomes clear that subsidies fail to offer the same robustness as tolls.
Ongoing work focuses on understanding how important robustness is and at what level of player heterogeneity will it become more advantageous to implement tolls than subsidies.

\bibliographystyle{IEEEtran}
\bibliography{../../../../library}

\appendix
We prove \cref{lem:scaling} using the definition of the Nash flow, and by showing this transformation does not affect user preferences.

\noindent \emph{Proof of \cref{lem:scaling}:}

Let $f^\prime$ be a Nash flow for a game $G \in \gee$ under influencing mechanism $T$.
User $x \in N_i$ observes cost 
\begin{equation}
J_x(P_x,f^\prime) = \sum_{e \in P_x} \ell_e(f^\prime_e) + \tau_e(f^\prime_e),
\end{equation}
and by the definition of Nash flow, will have preferences satisfying
\begin{equation}\label{eq:Nash}
J_x(P_x,f^\prime) \leq J_x(P^\prime,f^\prime)), \quad \forall P^\prime\in \paths_i.
\end{equation}
In the same flow $f^\prime$, but now under influencing mechanism $\hat{T}$, user $x$ will observe cost
\begin{align}
\hat{J}_x(P_x,f^\prime) &= \sum_{e \in P_x} \ell_e(f^\prime_e) + \lambda\tau_e(f^\prime_e) + (\lambda-1)\ell_e(f^\prime_e), \\
		&= \sum_{e \in P_x} \lambda(\tau_e(f^\prime_e) + \ell_e(f^\prime_e)) \\
		&= \lambda J_x(P_x,f^\prime).
\end{align}
\vs
Observe that through the same process, it can be shown that $\hat{J}_x(P,f^\prime) = \lambda J_x(P,f^\prime)$ for every $P \in \paths_i$.
From \eqref{eq:Nash},
\begin{align}
(1/\lambda)\hat{J}_x(P_x,f^\prime) &\leq (1/\lambda)\hat{J}_x(P_x,f^\prime), \quad \forall P_x \in \paths_i\\
\hat{J}_x(P_x,f^\prime) &\leq \hat{J}_x(P_x,f^\prime), \quad \forall P_x \in \paths_i. \label{eq:lamNash}
\end{align}
\eqref{eq:lamNash} holds for all $x\in N$, satisfying that $f^\prime$ is a Nash equilibrium in $G$ under $\hat{T}$.
It is therefore the case that any equilibrium in any game $G \in \gee$ under $T$ is also an equilibrium under $\hat{T}$, thus
\begin{equation}
\Lnash(G,T) = \Lnash(G,\hat{T}),
\end{equation}
and, because this holds for every game $G \in \mathcal{G}$, it certainly holds for the supremum over the set which is the same as \eqref{eq:lem_scale} by definition.
$\qed$

\vspace{0.1in}
\noindent \emph{Proof of \cref{lem:hetero}:}

First, we assume without loss of generality, that $\bmin=1$.
To see this, we make an equivalent problem where this is true and show the same price of anarchy bound will hold.
Let $T$ be any incentive mechanism and $\sdist$ be a family of sensitivity distributions with lower bound $\bmin$ and upper bound $\bmax$.
In any game $G \in \gee$, a player $x \in N_i$ will observe costs as expressed in \eqref{eq:player_cost}.
Observe that if we normalize every sensitivity distribution $s \in \sdist$ by multiplying by $1/\bmin$ and correspondingly scale the incentive by $\bmin$ the player cost remains unchanged.
It is therefore the case that any equilibrium is preserved and unchanged, enforcing that
\be
	\poa(\gee,\sdist,T) = \poa\left(\gee,\sdist/\bmin,\bmin \cdot T\right).
\ee
Accordingly, we will consider that $\bmin = 1$ throughout.

Let $f$ be a flow in $G\in \gee$ under sensitivity distribution $s\in \sdist$ an incentive mechanism $T$ that assigns tolls $\tau_e^+$.
From \cref{lem:scaling} a nominally equivalent incentive mechanism can be found by using the transformation $\hat{T}^+(\ell_e;\lambda) = (\lambda-1)\ell_e+\lambda T(\ell_e)$, where choosing $\lambda$ sufficiently close to zero will cause $\hat{T}$ to be a subsidy mechanism.
We will show that for any $\lambda \in (0,1)$, the incentive mechanism $\hat{T}$ will perform worse than $T$ at the introduction of player heterogeneity.

Let $\hat{s}$ be a new sensitivity distribution such that 
\be \label{eq:g_sens}
	\hat{s}_x = g(s_x,\lambda)= \frac{s_x}{\lambda + s_x - s_x\lambda},
\ee
for all $x \in [0,1]$.
Now, consider an agent's cost in flow $f$ with sensitivity $\hat{s}$ under incentive mechanism $\hat{T}$.
An agent $x \in N_i$ utilizing path $P_x$ in $f$ experiences cost,
\begin{align*}
		\hat{J}_x(P_x,f) &= \sum_{e \in P_x} \ell_e(f_e) + \hat{s}_x\hat{T}(\ell_e(f_e);\lambda)\\
		&= \sum_{e \in P_x} \ell_e(f_e) + \hat{s}_x[(\lambda-1)\ell_e+\lambda \tau_e^+(f_e)]\\
		&= \frac{\lambda}{\lambda + s_x - s_x\lambda} \sum_{e \in P_x}(\ell_e(f_e)+s_x\tau_e(f_e)),
\end{align*}
which is proportional to $J_x(P_x,f)$. By observing proportional costs, players preserve the same preferences over paths, preserving the same Nash flows.

Finally, we show that $\hat{s}$ is a feasible sensitivity distribution in $\sdist$.
From the original bounds $\bmin$ and $\bmax$, any generated distribution $\hat{s}$ will exist between $g(\bmin,\lambda)$ and $g(\bmax,\lambda)$.
From before, $\bmin = 1$, thus from \eqref{eq:g_sens}, $g(\bmin=1,\lambda) = 1 = \bmin$, for any $\lambda \in (0,1)$.
Now, observe that any generated distribution will satisfy
\be
	g(\bmax,\lambda) = \frac{\bmax}{\lambda+\bmax-\bmax \lambda} \leq \bmax,
\ee
for any $\lambda \in (0,1)$.
Thus any generated distribution $\hat{s}$ is sufficiently bounded by $\bmin$ and $\bmax$ and is a feasible distribution in $\sdist$.
By choosing $f$ to be a Nash flow, we can see that any Nash flow that can be induced by some $s \in \sdist$ while using $T$ can similarly be induced by $\hat{s} \in \sdist$ while using $\hat{T}$.
It is therefore the case that the price of anarchy with user heterogeneity will increase as $\lambda$ decreases, showing the monotonicity.
$\qed$

%

\end{document}